\documentclass{article}

\usepackage{amsthm}
\usepackage{amsmath}
\usepackage{amsfonts}
\usepackage{amssymb}
\usepackage{graphicx}
\usepackage[section]{algorithm}
\usepackage{algorithmic}

\newtheorem{theorem}{Theorem}[section]
\newtheorem{lemma}{Lemma}[section]

\DeclareGraphicsExtensions{.eps}

\begin{document}

\author{Asghar Asgharian Sardroud \thanks{Email: asgharian@aut.ac.ir}
 \and  Alireza Bagheri \thanks{Email:ar\_bagheri@aut.ac.ir}\\
 Department of Computer Engineering \& IT,\\ Amirkabir University of Technology, Tehran, Iran.
 }
\title{Unit-length embedding of cycles and paths on grid graphs}


\maketitle
\abstract{
Although there are very algorithms for embedding graphs on unbounded grids,
only few results on embedding or drawing graphs on restricted grids has been published.
In this work, we consider the problem of embedding paths and cycles on grid graphs.
We give the necessary and sufficient conditions for the existence of cycles of given length $k$ and paths of given length $k$ between two given vertices in $n$-vertex rectangular grid graphs and
introduce two algorithms with running times O$(k)$ and O$(k^2)$ for finding respectively such cycles and paths.
Also, we extend our results to $m\times n\times o$ 3D grids.
Our method for finding cycle of length $k$ in rectangular grid graphs also introduces a
linear-time algorithm for finding cycles of a given length $k$ in hamiltonian solid grid graphs.

}

\section{Introduction}\label{sec:intro}

In the area of graph drawing, a lot of research has been done on drawing or embedding various classes of graphs  on grids. Most of the research thus far, has considered unbounded grids, and in some cases tried to minimize the area used for drawing the given graph, and
only a few results deal with drawing or embedding graphs on restricted grids.
Felsner et al. \cite{felsner03straight} considered the problem of drawing graphs on $k$-tracks, i.e. $\infty\times k$ grids. They characterized  the 2-track drawable trees and devised a linear-time algorithm for constructing such drawings.
Cornelsen et al. \cite{cornelsen2006drawing} proposed a polynomial-time algorithm for
planar drawing of an arbitrary graph on two parallel lines.
Their results show that, only outer planar graphs may have such a drawing.
Suderman studied the various types of drawings of a tree on $\infty\times k$ grids \cite{suderman2004pathwidth}.
Using the concept of pathwidth, he worked out the optimal lower and upper bounds on the minimum value of $k$ for a tree in order to be drawable on a $\infty\times k$ grid.

In this paper, we consider unit-length drawing of cycles and paths in bounded grids, i.e. in grid graphs.
Grid graphs are finite vertex-induced subsets of infinite integer grids and they are widely used
in various areas of computer science
with applications in
VLSI layout design (e.g. \cite{hu2001survey,kastner2000predictable}), robotics and path planning (e.g. \cite{arkin2001optimal,arkin2000approximation}), computer graphics, overlay networks and parallel processing.

The problem of finding paths and cycles of a given length is the general case of the problems of finding Hamiltonian and longest paths and cycles.
These well-known problems are NP-hard for both general graphs \cite{diestel2010graph,garey1979computers} and grid graphs \cite{itai1982hamiltonian}.
However, there are polynomial-time algorithms for these problems on some subclasses of grid graphs. For example,
Umans et al. \cite{umans1997hamiltonian} proposed a polynomial-time algorithm for finding Hamiltonian cycles in solid grid graphs. Solid grid graphs are a subset of grid graphs which have no holes, i.e. any unbounded face in these graphs has unit area, but the polynomiality of Hamiltonian path, longest path and longest cycle problems are still open problems for this class of graphs.
Recently, Zhang and Liu \cite{zhang2011approximating} proposed a  $\frac{5}{6}$-approximation algorithm for the longest path problem in grid graphs that possess square free 2-factors.
A square free 2-factors is a 2-regular spanning subgraph in which any component has more than 4 vertices.
In \cite{salman2003spanning} salman defined 26 classes of grid graphs called alphabet graphs and determined which of these classes of graphs has Hamiltonian cycles.

Itai et al. \cite{itai1982hamiltonian} proposed a linear-time algorithm for finding a Hamiltonian path between two given vertices of a rectangular grid graph.
Rectangular grid graphs are grid graphs with the vertex set of all integer-coordinate points inside an axis-parallel rectangle.
For this problem, Chen et al. \cite{dong2002efficient} introduced a parallel algorithm that can find a Hamiltonian path between two given vertices of a rectangular grid graph in constant time, if there is a processor for each vertex of the graph.
Keshavarz et al. \cite{keshavarz2010longestingrid} considered the longest path problem in rectangular grid graphs and introduced a linear-time algorithm for finding a longest path between two given vertices in a rectangular grid graph.

The problems of Hamiltonian path, Hamiltonian cycle, longest path and  longest cycle are widely studied in the case of grid graphs.
Yet, we are not aware of any results on finding paths and cycles of a given length in these graphs.
We present a
linear-time algorithm for constructing a cycle of a given length in rectangular and Hamiltonian solid grid graphs, and a polynomial-time algorithm for finding a path of a given length between two specific vertices of a rectangular grid graph.

The remainder of this paper is organized as follows. Section \ref{sec:preli}, contains some necessary concepts and definitions. In Sections \ref{sec:cycle-k} and \ref{sec:path-k}, the necessary and sufficient conditions are given for the existence of respectively cycles of length $k$ and paths of length $k$ between two given vertices in a rectangular grid graph.
Although the proofs given in Sections \ref{sec:cycle-k} and \ref{sec:path-k} are constructive and implicitly give algorithms for finding desired cycles or paths, since these algorithms are not efficient, Section \ref{sec:algs} introduces two efficient algorithms for these problems.
In Section \ref{sec:3d}, we extend our results to 3D grids and
Section \ref{sec:concl} concludes the paper.

\section{Preliminaries} \label{sec:preli}

The \emph{two-dimensional integer grid} $G^\infty$ is an infinite graph
with the vertex set of all the points of the Euclidean plane that have integer coordinates.
In this graph, there is an edge between any two vertices of unit distance.
A \emph{grid graph} is a finite vertex-induced subgraph of $G^\infty$.
For a vertex $v$ of $G^\infty$, we use $x(v)$ and $y(v)$ to denote
the $x$ and $y$ coordinates of its corresponding point in the Euclidean plane.
Furthermore, we color $v$ \emph{white} if $x(v)+y(v)$ is even; otherwise, we color it \emph{black}.
Grid graphs are 2-colorable, because there is no edge between white and black vertices in these graphs.
For a grid graph $G_g$, its \emph{size} $|G_g|$ is considered to be the number of its vertices. Similarly
for a path or cycle $X$, its \emph{length} $|X|$ is equal to the number of its vertices.

In a 2-colorable graph any cycle or path
alternates between black and white vertices, so we have the following lemma.

\begin{lemma} \label{lem:odd-even-len}
In a grid graph, any cycle and any path between two different-colored vertices has even length,
and any path between two same-colored vertices has odd length.
\end{lemma}

Grid graph $G_g$ is \emph{solid} if $G^\infty \setminus G_g$ is connected,
where $G^\infty \setminus G_g$ denotes the graph obtained from $G^\infty$ by removing vertices
of $G_g$ together with their incident edges.
A \emph{rectangular grid graph}, $R(m,n)$
(or $R$ for short), is a solid grid graph whose vertex set is
$V(R)=\{v \ |\ 1 \leq x(v)\leq m,\ 1\leq y(v)\leq n\}$.
Similarly, 3D grid graph $R(m,n,o)$ can be defined by the vertex set $V(R)=\{v \ |\ 1 \leq x(v)\leq m,\ 1\leq y(v)\leq n, 1 \leq z(v)\leq o\}$.

In the following, $R$ denotes the given rectangular grid graph in which we would like to find path $P$ between two vertices $s$ and $t$, or cycle $C$ of the given length $k$.
In the figures, the vertices are illustrated by bullets, the edges of $P$ or $C$
by solid lines and the other edges of $R$ by dotted lines.
The paths are assumed to be directed from $s$ to $t$ and the cycles are assumed to be directed in clockwise direction.
Based on the direction of the edges, we have four type of edges, i.e. \emph{up}, \emph{down}, \emph{left} and \emph{right} edges.

Our main idea for finding paths or cycles of length $k$ is to start from a longest path or cycle and decrease its length to $k$ by contracting some parts of it called \emph{caves}.
A cave $c$ in a path $P$ or cycle $C$ is a minimal sub path of $P$ or $C$  such that its first and last edges are in opposite directions (see Figure \ref{fig:cave}(a)).
We use letters $p$ and $q$ to refer to the start and end vertex of cave $c$ respectively. All the edges and vertices between $p$ and $q$ (i.e. the edges and the vertices of $R$ in the shortest path between $p$ and $q$ excluding $p$ and $q$) are considered to be \emph{inside} the cave.
A cave is called \emph{contractible} if none of the vertices inside it are in $P$ or $C$.
Figure \ref{fig:cave}(a) shows part of a rectangular grid graph containing a non-contractible cave between two vertices $p$ and $q$. \emph{Contracting} a contractible cave is done by replacing the cave with the shortest path between $p$ and $q$ as shown in Figure \ref{fig:cave}(b).
The cave $c$ of cycle $C$ is \emph{convex cave} if the edges and vertices inside $c$ also lie inside $C$; otherwise, it is a \emph{concave cave}.
\begin{figure*}
\begin{center}
\footnotesize\centering
\centerline{
\includegraphics[scale=1]{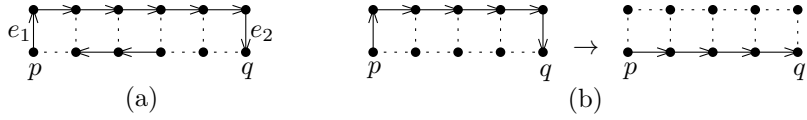}}
\caption{(a) A non-contractible cave (b) Contracting a cave}
\label{fig:cave}
\end{center}
\end{figure*}%
\section{Cycles of length $k$} \label{sec:cycle-k}
In this section, we give the necessary and sufficient conditions for the existence of cycles of length $k$ in $R$.
First, we find the length of the longest cycles of a grid graph $R$. Then using the concept of contractible convex caves we show that for any even $k>4$ that is not larger than the length of the longest cycles of $R$, there is a cycle of length $k$ in $R$.
The following lemma gives the length of the longest cycles of $R$, and shows how we can find such cycles.
\begin{lemma} \label{lem:longest-cycle}
The length of the longest cycles of a rectangular grid graph $R(m,n)$ is $m\times n$ when $m\times n$ is even, and $m \times n -1$ otherwise.
\end{lemma}
\begin{proof}
If $m\times n$ is even then $m$ or $n$ is even and we can find a cycle of size $m\times n$ as shown in Figure \ref{fig:longest-cycle}(a). If $m\times n$ is odd then both of $m$ and $n$ are odd and we can find a cycle of size $m\times n-1$ as shown in Figure \ref{fig:longest-cycle}(b). Because $R$ is 2-colorable, any cycle in $R$ has even length so these cycles are the longest.
\end{proof}

\begin{figure}[tb]
\begin{center}
\footnotesize\centering
\centerline{
\includegraphics[scale=.8]{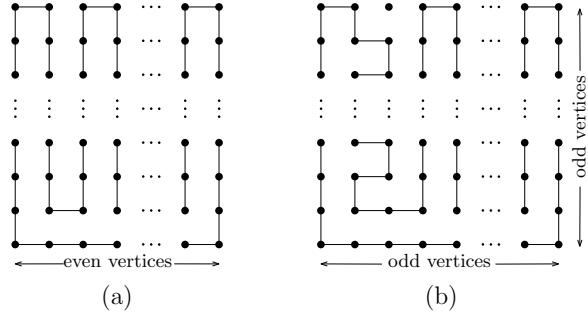}}
\caption{ Pattern of a longest cycle in (a) even-sized and (b) odd-sized rectangular grid graphs}
\label{fig:longest-cycle}
\end{center}
\end{figure}

Any cycle $C$ in a solid grid graph $G_g$ can be considered a simple orthogonal polygon. The vertices of this orthogonal polygon are those vertices of $C$ in which the direction of the edges of $C$ change.
An \emph{orthogonal ear} in an orthogonal polygon is four consecutive vertices $v_1,v_2,v_3$ and $v_4$ such that the line segment connecting $v_1$ to $v_4$ lies completely inside the polygon.
Two orthogonal ears are  non-overlapping, if the insides of the two quadrilaterals defined by these ears are disjoint.
\begin{lemma} \label{lem:cycle-cave}
Let $C$ be a cycle in a solid grid graph $G_g$, and $e$ be one of its edges. Then $C$ contains a contractible convex cave $c$ not containing $e$.
\end{lemma}
\begin{proof}
Without loss of generality we assume that the orthogonal polygon corresponding to $C$ has more than four vertices, otherwise the lemma is trivial.
We know that every orthogonal polygon with at least five vertices has at least two non-overlapping orthogonal ears \cite{orourke1987art}.
Any orthogonal ear, in the orthogonal polygon, corresponds to a contractible convex cave in $C$. Thus, $C$ should contain at least two contractible convex caves, and one of these two caves does not contain the edge $e$.
\end{proof}

\begin{theorem} \label{thm:cycle-k}
 In a rectangular grid graph $R(m,n)$ there is a cycle of length $k$ if and only if $m,n>1$, $k$ is even and $4\leq k \leq m \times n$.
\end{theorem}
\begin{proof}
Clearly if $m=1$ or $n=1$ there is no cycle in $R$. Because $R$ is 2-colorable any cycle in $R$ should have even length. Furthermore, based on Lemma \ref{lem:cycle-cave} any cycle in $R$ has a contractible cave such that its contraction reduces the length of the cycle by two.
Therefore starting from the longest cycle of $R$ (obtained by Lemma \ref{lem:longest-cycle}) and by repeatedly finding and contracting such caves we can construct a cycle of length $k$ for any even $k$ satisfying $4\leq k \leq m \times n$. Note that there is no cycle of length less than four in $R$.
\end{proof}

Lemma \ref{lem:cycle-cave} has more general consequences. For example, this lemma shows that in a Hamiltonian solid grid graph with $n$ vertices there is a cycle of length $k$ for any even $k$ satisfying $4\leq k\leq n$.

\section{Paths of length $k$}\label{sec:path-k}
In this section, we use a similar approach for giving the necessary and sufficient conditions for the existence of paths of length $k$ between two given vertices in a rectangular grid graph. We show that any path between $s$ and $t$ that is not a shortest path has a contractible cave. Thus, starting from a longest path between $s$ and $t$ we can construct a path of any desired length.

A path from $s$ to $t$ is \emph{monotone} if it does not contain two edges in opposite directions i.e. it does not contain both up and down or both left and right edges.
\begin{lemma}\label{lem:monotone-shortest}
In a rectangular grid graph, a path between any two vertices $s$ and $t$ is shortest if and only if it is monotone.
\end{lemma}
\begin{proof}
Without loss of generality assume that $x(s) \leq x(t)$ and $y(s) \leq y(t)$. Clearly, any path between $s$ and $t$ must contain at least $y(t)-y(s)$ up edges and at least $x(t)-x(s)$ right edges. Furthermore,
any monotone path between $s$ and $t$ should contain only up and right edges, so it can not contain more than
$y(t)-y(s)$ up edges and more than $x(t)-x(s)$ right edges. Therefore, any monotone path between $s$ and $t$ is a shortest path. On the other hand, any shortest path has only $y(t)-y(s)$ up edges and $x(t)-x(s)$ right edges so is monotone.
\end{proof}

In the following, we are going to show that any non-monotone path in $R$ has at least one contractible cave.
Clearly any non-monotone path contains at least one cave because it contains at least two edges in opposite directions.
Let $c$ be a non-contractible cave and $p$ and $q$ be its end vertices.
Suppose that edges inside $c$ incident to $p$ or $q$ are not in $P$ (e.g. $e_1$ and $e_2$ in Figure \ref{fig:c-cave}), because in this case we can easily find a contractible cave with three edges and contract it.
Then $c$ should be blocked by $s$ or $t$ (i.e. $s$ or $t$ are inside $c$) or by some smaller caves.
A cave $c$ is \emph{blocked} by cave $c'$ if all the edges of $c'$ are inside $c$ except its first and last edges.
Figure \ref{fig:c-cave} shows a cave blocked by $t$ and two smaller caves. Note that, edges like $e_3$ and $e_4$ should be in $P$ because each vertex of $P$, except $s$ and $t$, should have two incident edges in $P$.
\begin{figure}
\begin{center}
\footnotesize\centering
\centerline{
\includegraphics[scale=1]{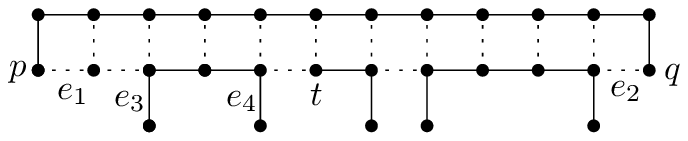}}
\caption{A non-contractible cave blocked by two smaller caves and $t$}
\label{fig:c-cave}
\end{center}
\end{figure}

\begin{lemma} \label{lem:con-cave}
Any non-monotone path in a rectangular grid graph contains at least one contractible cave.
\end{lemma}
\begin{proof}
Let $P$ be a non-monotone path between vertices $s$ and $t$ in a rectangular grid graph $R$.
Then it should contain at least one cave. Let $c$ be the nearest cave to $s$ (i.e. when we start from $s$ and go along $P$, $c$ is the first cave that we encounter). We assume that edges of $R$ inside $c$ and incident to the end vertices of $c$ are not in $P$, otherwise it would be easy to find a contractible cave with three edges. Hence $c$ should be blocked by $s$ or $t$ or some other smaller caves. We can search the blocking caves of $c$ recursively to find a contractible cave. Since the length of blocking caves of a cave is smaller than its length, caves cannot block each other in a cyclic order. So we should finally find a contractible cave except in the case that there is a sequence of caves $c=c_1,c_2,...,c_k$ where $c_i (1\leq i<k)$ is blocked only by $c_{i+1}$ and $c_k$ is blocked only by $t$ (see Figure \ref{fig:b-cave}).
Note that neither $c$ nor its directly or indirectly blocking caves can be blocked by $s$, because we assumed that $c$ is the nearest cave to $s$.

In this case, let 
$v$ be the vertex of $c_k$ adjacent to $t$.
The sub path of $P$ between $v$ and $t$ together with edge $(t,v)$ makes cycle $Q$.
If $s$ is inside $Q$ (Figure \ref{fig:p-cave}(b)) then clearly connecting $s$ to any vertex of $P$ with only one edge of $R$ can not result in a cycle $Q'$ containing $t$ in its inside.
Therefore, without loss of generality we can assume that $s$ is not inside $Q$ (Figure \ref{fig:p-cave}(a)), otherwise we can change the direction of the path $P$ and exchange the roles of $s$ and $t$. Using Lemma \ref{lem:cycle-cave}, we can find a contractible cave in $Q$ not containing edge $(t,v)$. This completes our proof.
\end{proof}

\begin{figure}
\begin{center}
\footnotesize\centering
\centerline{
\includegraphics[scale=1]{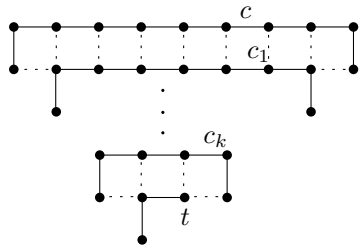}}
\caption{A non-contractible cave $c$ such that none of its directly or indirectly blocking caves are contractible}
\label{fig:b-cave}
\end{center}
\end{figure}

\begin{figure*}
\begin{center}
\footnotesize\centering
\centerline{
\includegraphics[scale=1]{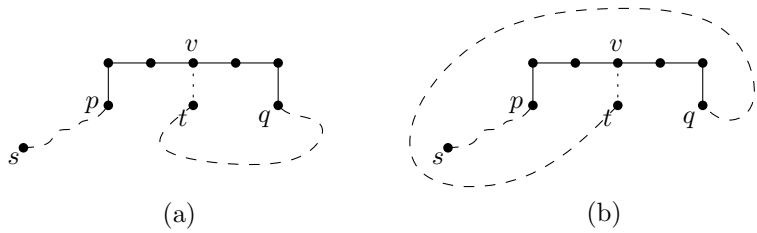}}
\caption{$s$ is in the (a) outside and (b) inside of the cycle obtained by connecting $t$ to $v$ }
\label{fig:p-cave}
\end{center}
\end{figure*}

\begin{theorem}\label{thm:path-k}
Let $s$ and $t$ be two vertices of a rectangular grid graph $R$. There is a path of length $k$ between $s$ and $t$ in $R$ if and only if $l\leq k\leq L$ and $k=l\ (mod\ 2)$, where $l$ and $L$ are respectively lengths of the shortest and the longest paths between $s$ and $t$ in $R$.
\end{theorem}

\begin{proof}
Clearly there is no path of any length shorter than $l$ or longer than $L$ between $s$ and $t$. Based on Lemma \ref{lem:odd-even-len} either all of the paths between $s$ and $t$  have even lengths or all of them have odd lengths, so $l\leq k\leq L$ and $k = l\ (mod\ 2)$ are necessary conditions for the existence of a path of length $k$.
For the sufficiency of these conditions note that there is a linear-time algorithm for finding a longest path between $s$ and $t$ in $R$ \cite{keshavarz2010longestingrid}. Furthermore if a path between $s$ and $t$ is not the shortest one, based on Lemma \ref{lem:monotone-shortest} it is non-monotone so by the Lemma \ref{lem:con-cave} it has at least one contractible cave. By contracting this cave we can reduce its length by two. Thus starting from the longest path, and repeatedly finding and contracting a contractible cave we can finally construct a paths of length $k$, if $k$ satisfies the mentioned conditions.
\end{proof}

\section{The algorithms}\label{sec:algs}
We demonstrated the necessary and sufficient conditions for the existence of cycles or paths of length $k$ in a rectangular grid graph.
The proof of the Theorem \ref{thm:cycle-k} is constructive and provides an algorithm for finding a cycle of length $k$ in a solid grid graph $G_g$.
However, it requires O$(|G_g|^2)$ time to execute. Because a longest cycle of $G_g$ can be found in O$(|G_g|)$ time and its length is O$(|G_g|)$.
Similarly, the constructive proof of the Theorem \ref{thm:path-k} provides an algorithm for finding a path of length $k$ between two vertices of a rectangular graph $R$ which requires O$(|R|^2)$ time to execute.
In this section we will present two efficient algorithms for finding such cycles and paths.

First, we show that reducing the length of a cycle $C$ in a solid grid graph $G_g$ to the desired length can be done in time O$(|G_g|)$.
The procedure \emph{ShrinkCycle} in the algorithm \ref{alg1} reduces the length of $C$ by an even integer $0<i \leq |C|-4$, without removing a given edge $e$ from $C$, in O$(|G_g|)$ time.
In this algorithm, $C_{u,v}$ refers to the sub path of $C$ starting from $u$ going in the clockwise direction around $C$ and ending at $v$, where $u$ and $v$ are two vertices of $C$.
The following lemma proves the correctness and time complexity of Algorithm \ref{alg1}.
\newcommand{\LET}{\STATE \textbf{let }}
\newcommand{\SET}{\STATE \textbf{set }}
\newcommand{\PROC}{\STATE \textbf{procedure }}
\begin{algorithm*}[t]
\caption{Reducing length of a cycle}
\label{alg1}
\begin{algorithmic}
\PROC ShrinkCycle($C$,$e$,$i$)
\STATE /* $C$ is a cycle in a solid grid graph, and $e$ is a given edges of $C$ */
\STATE /* $i$ is an even integer such that $0\leq i \leq |C|-4$ */
\end{algorithmic}
\begin{algorithmic}[1]
\LET $e$=$(t,s)$, and $s$ is next to $t$ in clockwise order around $C$
\WHILE{$i > 0$}
\STATE starting from $s$ scan $C$ in clockwise order to find the first convex cave $c$ not containing $e$
    \IF{$c$ is contractible}
        \STATE contract $c$
        \STATE $i \leftarrow i-2$
    \ELSE
        \LET $p$ be the start vertex of $c$ around $C$ in clockwise order
        \LET $v$ be the vertex of $C$ inside $c$ closest to $p$
        \LET $u$ be the adjacent vertex of $v$ in $c$ different from $p$
        \IF{$|C_{u,v}|-2 \leq i$}
            \STATE Remove $C_{u,v}$ from $C$ and add edge $(v,u)$ to $C$
            \STATE $i \leftarrow i-(|C_{u,v}|-2)$
        \ELSE
            \LET $Q$ be the cycle formed by $C_{u,v}$ and edge $(v,u)$
            \STATE ShrinkCycle($Q$,$(v,u)$,$i$)
        \ENDIF
    \ENDIF
\ENDWHILE
\end{algorithmic}
\end{algorithm*}

\begin{lemma}\label{lem:cycle-shrink}
Let $C'$ be a cycle in a solid grid graph $G_g$, and $e$ be a given edge of $C'$.
Then we can find a cycle $C$ of length $k<|C'|$ containing $e$ in O$(|G_g|)$ time, if it exists.
\end{lemma}
\begin{proof}
We prove the lemma by showing that the procedure \emph{ShrinkCycle} of Algorithm \ref{alg1} works correctly and runs in time O$(|G_g|)$,
when its preconditions is satisfied (i.e. its parameter $C$ is a cycle in a solid grid graph and $i$ is an even integer and $0 \leq i\leq |C|-4$).
By calling the \emph{ShrinkCycle} with parameters $C'$, $e$ and $i=|C'|-k$ we can construct the desired cycle $C$.
In this procedure, if $|C|=4$ then $i$ should be zero, so we can assume that $|C|>4$.
Based on the Lemma \ref{lem:cycle-cave}, there is at least one contractible convex cave not containing $e$ on $C$.
Hence, in line 3, $c$ always can be found.
If $c$ is contractible, then its contraction reduces the length of $C$ by two, and if it is not contractible, its general situation is illustrated in Figure \ref{fig:general-convex-cave}.
In this case, let $v$ be the vertex of $C$ inside $c$ closest to $p$ and $u$ be the adjacent vertex of $v$ in $c$ different from $p$. If $|C_{u,v}|-2 \leq i$, then line 12 of the procedure,
by removing $C_{u,v}$ from $C$ and adding edge $(u,v)$ to it, decreases its length by $|C_{u,v}|-2$.
Otherwise, let $Q$ be the cycle formed by sub path $C_{u,v}$ nad edge $(v,u)$. The procedure recursively reduces the length of cycle $Q$ by $i$ at line 16.
Clearly, at this line $(v,u)$ is an edge of $Q$ and $0\leq i \leq |Q|-4$ so the preconditions of the procedure \emph{ShrinkCycle} are satisfied.
By induction on the length of $C$, we can assume that the recursive call at line 16 works correctly.
Because, $|Q|<|C|$ and when $|C|=4$ ($i=0$), the procedure trivially works correctly.
Therefore, Algorithm \ref{alg1} works correctly.

The \emph{while} statement of line 2 of the procedure can iterate at most O$(|C|)$ times (including the iterations of the recursive calls), because in each iteration either $i$ or $|C|$ decreases.
Storing the track of the scanned part of $C$, we can find all the consecutive convex caves of a cycle in one scan.
In other words, let $w$ be the last scanned vertex of $C$ and let $v_1=s,v_2,...,v_l=w$ be the ordered list of vertices of $C_{s,w}$ in which direction of the edges change, together with $s$ and $w$.
For each consecutive pair of vertices of this list, $v_j$ and $v_{j+1}$, let $V_j$ be the set of vertices of $C$ lying to the right of the line segment connecting $v_j$ and $v_{j+1}$ and adjacent to some vertices of $C_{v_j,v_{j+1}}$. Also let $E_j$ be the set of edges of $C$ whose end vertices are in $V_j$. We can construct and maintain these lists while scanning $C$ (starting from $s$) and contracting some of its caves. Using vertices $v_{l-3},..,v_l$ and the sets $V_{l-2}$ and $E_{l-2}$ one easily can check that if there is a
cave ending at $w$, and whether it is contractible or not in time O$(1)$.
Hence, the total time spent at lines 3 and 4 of the procedure (including the recursive calls) is O$(|C|)$.

The total time spent at line 5 can not exceed O$(|G_g|)$, because contracting a cave $c$ requires O$(|c|)$ time and when $c$ is contracted, the number of vertices of $G_g$ lying inside $C$ is decreased by $|c|-2$, and it is never increased again.
Other lines of this procedure can be implemented in time O$(1)$.
Therefore, overall running time of the algorithm is O$(|G_g|)$.
\end{proof}


\begin{figure}
\begin{center}
\footnotesize\centering
\centerline{
\includegraphics[scale=1]{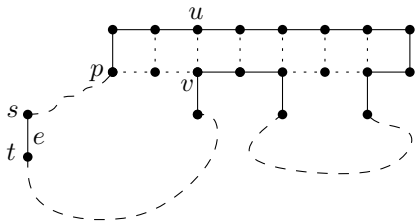}}
\caption{A general case of a non-contractible convex cave in a cycle}
\label{fig:general-convex-cave}
\end{center}
\end{figure}

The size of $R$ can be arbitrarily larger than $k$. So applying the Algorithm \ref{alg1} to a longest cycle of $R$ may require more than O$(k)$ time.
However, instead of $R$, we can use a rectangular subgraph of $R$ having O$(k)$ size.
For example, for a rectangular grid graph $R(m,n)$ with $m\geq n$ and $m\times n > k$, let $R'$ be the induced rectangular subgraph of $R(m,n)$ with vertex set $V'$, where:

\noindent  \small
$V'=\begin{cases}
     \text{\small{$ \{v| \ 1 \leq x(v),y(v) \leq \lceil \sqrt{k+1} \rceil \} $}}, &\text{\small{if $n\geq \lceil \sqrt{k+1} \rceil\text{,} $}} \\
     \text{\small{$ \{v|\ 1 \leq x(v) \leq n, 1\leq y(v),$}} \\ \text{\small{$\ \ y(v) \leq \text{min}(m,\lceil(k+1)/n \rceil)\}$}},& \text{otherwise.}
      \end{cases}$
Clearly, $R'$ has a cycle of length $k$ and $|R'|$ is at most \small{$\lceil \sqrt{k+1} \rceil^2 = k+2\sqrt{k+1}+3$}.

\begin{theorem}\label{thm:cycle-time}
In a rectangular grid graph $R$, a cycle of length $k$ can be found in O$(k)$ time, if it exists.
\end{theorem}
\begin{proof}
First we find a rectangular subgraph $R'$ of $R$ of size O$(k)$ having more than $k$ vertices.
Then by Lemma \ref{lem:cycle-shrink} we can construct a cycle of length $k$ from the longest cycle of $R'$ in O$(k)$ time.
\end{proof}

Algorithm \ref{alg1} also can be used to shrink cycles in solid grid graphs. For example, Figure \ref{fig:ex-fig}(b) shows a cycle of length 58 on a solid grid graph which is obtained from its Hamiltonian path (Figure \ref{fig:ex-fig}(a)) by applying the procedure ShrinkCycle of Algorithm \ref{alg1}.
In the remaining of this section, we will show that a path of length $k$ between two vertices of a rectangular grid graph can be found in O$(k^2)$ time.

\begin{lemma}\label{lem:path-shrink}
In a rectangular grid graph $R$, having a path $P$ between two vertices $s$ and $t$, we can find a path of the given length $k<|P|$ between $s$ and $t$ in O$(|P|^2)$ time, if it exists.
\end{lemma}
\begin{proof}
By Lemma \ref{lem:con-cave}, while $P$ is not a shortest path it contains at least a contractible cave.
Such a cave can be found in O$(|P|)$ time, and we can contract at most O$(|P|)$ caves. Therefore, if a path of length $k$ between $s$ and $t$ exists, it can be found in O$(|P|^2)$ time.
\end{proof}

One may think it is easy to reduce the time complexity of lemma \ref{lem:path-shrink} to O$(|P|)$, similar to Lemma \ref{lem:cycle-shrink}.
But note that, while shrinking a cycle, the algorithm contracts only convex caves, so it sweeps the interior of the cycle only once.
However, there is no convex caves on paths, and contracting caves may cause some parts of the graph to be swept multiple times.
\begin{theorem}\label{thm:path-time}
In a rectangular grid graph $R$, a path of length $k$ between two given vertices $s$ and $t$ can be found in O$(k^2)$ time, if it exists.
\end{theorem}
\begin{proof}
We show how to construct a path $P$ between $s$ and $t$ such that $|P|\geq k$ and $|P|\in$ O$(k)$.
Then, by using $P$ and applying Lemma \ref{lem:path-shrink} we can construct a path of length $k$ in O$(k^2)$ time.
If there is a rectangular subgraph $R'$ of $R$ including $s$ and $t$ such that $|R'|> k+2$ and $|R'|\in$ O$(k)$, then we let $P$ be a longest path between $s$ and $t$ in $R'$ (see Figure \ref{fig:path-o-k}(a)).
Otherwise let $R'$ be a rectangular subgraph of $R$ including $s$ such that $|R'|> k+2$ and $|R'|\in$ O$(k)$ and let $v$ be a closest vertex of $R'$ to $t$ distinct from $s$.
Then, there is a longest path between $s$ and $v$ in $R'$ and a shortest path between $v$ and $t$ in $R$ such that these two paths have no edge in common. Let $P$ be the concatenation of these two paths (see Figure \ref{fig:path-o-k}(b)).
Because the length of a longest path between two vertices of a rectangular grid graph $R(m,n)$ is at least $m \times n -2$ \cite{keshavarz2010longestingrid}, in both cases $|P|>k$ and $|P| \in\text{O}(k)$. This completes our proof.
\end{proof}

\begin{figure}
\begin{center}
\footnotesize\centering
\centerline{
\includegraphics[scale=.9]{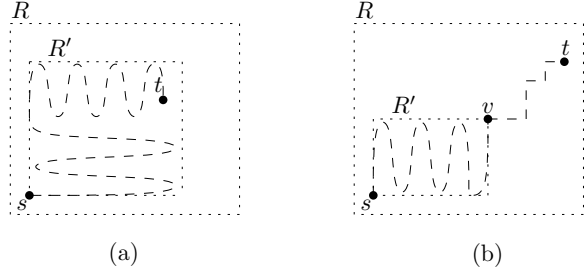}
}
\caption{In both (a) and (b) the dashed line shows the desired path of Theorem \ref{thm:path-time} }
\label{fig:path-o-k}
\end{center}
\end{figure}

\begin{figure*}
\begin{center}
\footnotesize\centering
\centerline{
\includegraphics[scale=.65]{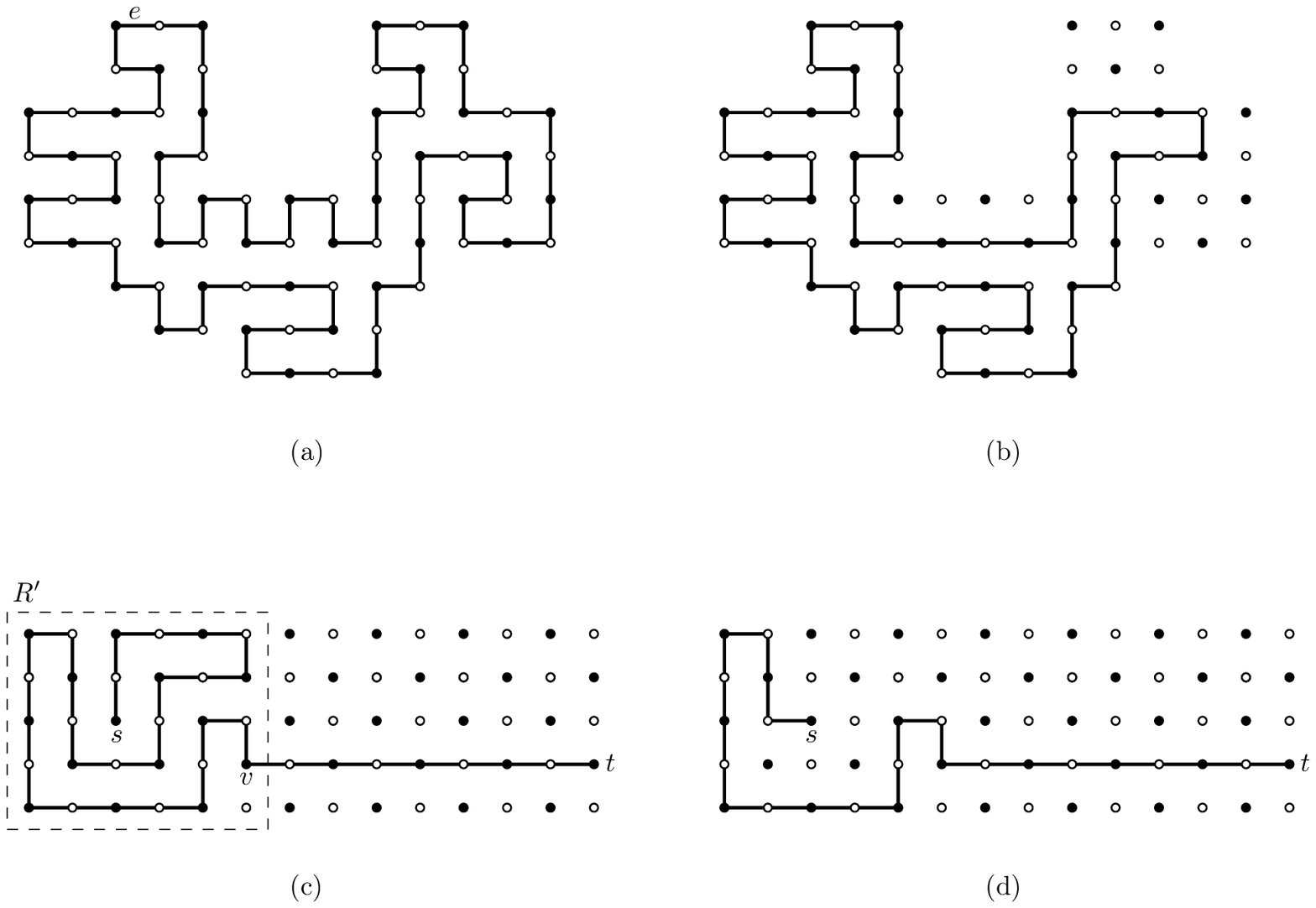}}
\caption{(a) a solid grid graph with its Hamiltonian cycle,  (b) a cycle of length 58 obtained by applying the procedure ShrinkCycle on the Hamiltonian cycle, (c) a path $P$ obtained by joining a longest path between $s$ and $v$ in $R'$ and the shortest path between $v$ and $t$, (d) a path of length 25 between $s$ and $t$ obtained from $P$ using Theorem  \ref{thm:path-time}}
\label{fig:ex-fig}
\end{center}
\end{figure*}
See Figures \ref{fig:ex-fig}(a) and (b) for examples.

\section{Extending the results to 3D grids}\label{sec:3d}
The problems of embedding length $k$ cycles and length $k$ paths
between two given vertices $s$ and $t$ in a  3D grid graph $R(m,n,o)$
can be solved using the introduced algorithms for solving these problems in rectangular grid grpaphs.

To this aim, consider a spanning subgraph of $R$ illustrated in Figure \ref{fig:3d}(b).
This subgraph is isomorphic to the rectangular grid graph $R'(m,no)$ using the bijective mapping
$F: V(R)\mapsto V(R')$ defined by the following:

\begin{figure*}
\begin{center}
\footnotesize\centering
\centerline{
\includegraphics[scale=.95]{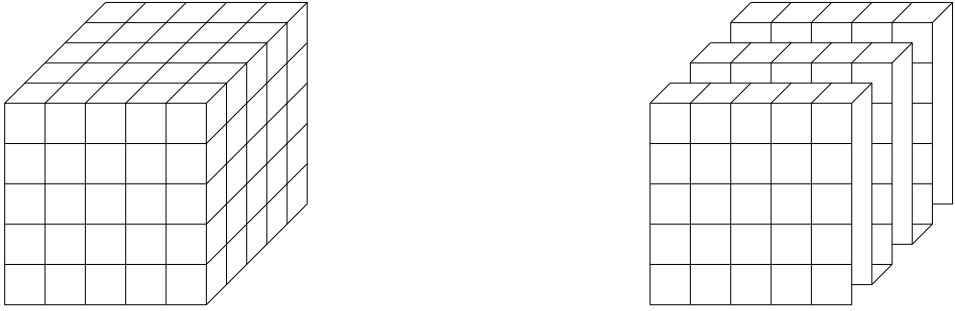}}
\caption{(a) a 3D grid graph $R(m,n,o)$ and (b) a subgraph of $R$ isomorphic to rectangular grid graph $R'(m,no)$ }
\label{fig:3d}
\end{center}
\end{figure*}

\begin{center}
\parbox[][][c]{.5\textwidth}{
$x(F(v)) = x(v),$

$y(F(v))=
\begin{cases}
n(z(v)-1)+y(v) & \mbox{if  $z(v)$ is odd,} \\
nz(v)-y(v)+1 & \mbox{if  $z(v)$ is even.}
\end{cases}$
}
\end{center}

Since, 3D grid graphs are 2-colorable as rectangular grid graphs,
the stated conditions for the existance of the desired cycles and paths in $R'$ are true also for $R$. Similarly,  the upper bounds proposed in \cite{keshavarz2010longestingrid} for the length of longest cycle and length of longest path between two given vertices in $R'$ is true also for $R$.
Thus, any longest cycle in $R'$ indicates a longest cycle in $R$ and any longest path between $F(s)$ nad $F(t)$ in $R'$ indicates a longest path between $s$ and $t$ in $R$.
Moreover, length of shortest cycles in both $R$ and $R'$ is four.
Thus, the mapping of any cycle of length $k$ in $R'$ using $F^{-1}$ is a cycle of length $k$ in $R$.
However, the same argument is not true about the paths of length $k$,
because the length of the shortest paths between $s$ and $t$ in $R$ may be less than the length of the shortests path between $F(s)$ and $F(t)$ in $R'$.

For finding a path of length $k$ between the vertices  $s$ and $t$ in $R$,
without loss of generality, let $z(t) \ge z(s)$ and
let $t_0$$=$$t, t_1,.. ,t_{z(t)-z(s)}$
be vertices of $R$ such that $x(t_i)=x(t)$,  $y(t_i)=y(t)$ and $z(t_i)=z(t)-i$
for $0\le i \le z(t)-z(s)$.
Also, without loss of generality, we can assume that $t_{z(t)-z(s)}$ is distinct from $s$,
 because as otherwise we can swap the roles of $y$ and $z$ cordinates in the defination of bijection $F$.

Let  $t_j$ be the first vertex among  $t_0, t_1,.. ,t_{z(t)-z(s)}$ such that
the length of the shortest path between $F(s)$ and $F(t_j)$ in $R'$ is not more than $k-j$. Such a $t_j$ exists if $k$ is not less than the length of the shortest path between $s$ and $t$, because the length of the shortest path between $F(s)$ and $F(t_{z(t)-z(s)})$ in $R'$ is equal to the length of the shortest path between $s$ and $t$ in $R$ minus $j$.
Now, if $P$ be a path of length $k-j$ between  $F(s)$ and $F(t_j)$ in $R$
(found by the algorithm proposed in Section \ref{sec:algs}) the concatenation
of the  mapping of $P$ using $F^{-1}$ and the sequence of vertices $t_{j-1},.. ,t_0$%
$=$$t$ results a path of length $k$ between $s$ and $t$ in $R$.
Note that, the mapping of $P$ using $F^{-1}$ will not contain any of vertices $t_{j-1},.. ,t_0$, as otherwise the $t_j$ is not the first vertex among  $t_0,.. ,t_{z(t)-z(s)}$ with the desired property which is a contradiction.

\section{Conclusions}\label{sec:concl}

In this paper we first give the necessary and sufficient conditions for the existence of cycles and paths of a given length $k$, between two given vertices in a rectangular grid graph. Using these conditions one can verify in O$(1)$ time whether such a path or cycle exits. We next introduced an algorithm for finding such cycles in optimal time O$(k)$, and an algorithm for finding such paths in O$(k^2)$ time. However, the problems of finding paths and cycles of a given length $k$ is still open for solid grid graphs.

\bibliographystyle{plain}
\bibliography{k-path}
\label{sec:biblio}

\end{document}